\documentclass[preprint,12pt]{elsarticle}

\usepackage{
amsmath,
amscd,
amssymb,
}
\usepackage{comment}
\usepackage{tikz}
\usetikzlibrary{positioning,arrows,calc}
\usepackage{xspace}

\usepackage{graphicx}

\usepackage{epstopdf}

\usepackage{url}

\newcommand{\Z}{\mathbb{Z}}
\newcommand{\N}{\mathbb{N}}

\newcommand{\ID}{\mathrm{id}}

\newcommand{\clone}[1]{\left\lceil#1\right\rceil}
\newcommand{\gen}[1]{\left\langle#1\right\rangle}
\newcommand{\id}{id}
\newcommand{\controlled}[2]{f_{#1,#2}}
\newcommand{\cp}[2]{CP(#1,#2)}
\newcommand{\Cone}{B(A)}
\newcommand{\Ctwo}{Cons(A)}
\newcommand{\Cthree}{Even(A)}
\newcommand{\Cfour}{ECons(A)}
\newcommand{\Ctwophi}{Cons^\phi(A)}
\newcommand{\Cfourphi}{ECons^\phi(A)}

\newcommand{\W}{\cal W}

\newcommand{\Sym}{\mathrm{Sym}}
\newcommand{\Alt}{\mathrm{Alt}}
\newcommand{\sym}{\mathrm{Sym}}
\newcommand{\alt}{\mathrm{Alt}}
\newcommand{\what}{revital}

\newcommand\xqed[1]{%
  \leavevmode\unskip\penalty9999 \hbox{}\nobreak\hfill
  \quad\hbox{#1}}

\newcommand{\new}{}

\newcommand{\UTU}{
 Department of Mathematics and Statistics, University of Turku, Finland.
 }

\newcommand{\LINZ}{
 Institute for Algebra, Johannes Kepler University Linz, Austria
and Time's Up Research, Linz, Austria.
 }

\newtheorem{theorem}{Theorem}
\newtheorem{lemma}[theorem]{Lemma}
\newtheorem{corollary}[theorem]{Corollary}
\newdefinition{definition}{Definition}
\newproof{proof}{Proof}

\journal{Theoretical Computer Science}

\begin{document}

\begin{frontmatter}

\title{Finite generating sets for reversible gate sets under general conservation laws\tnoteref{t1,t2}}

\tnotetext[t1]{The authors would like to acknowledge the contribution of the COST Action IC1405.
This work was partially funded by SFB Project F5004 of the Austrian Science
Foundation, FWF, by the Academy of Finland grant 296018, and by FONDECYT research grant 3150552.}
\tnotetext[t2]{A preliminary version of this work was presented at the 8th International Conference on Reversible Computation, RC 2016~\cite{rc2016}.}


\author[linz]{Tim Boykett}
\author[utu]{Jarkko Kari}
\author[utu]{Ville Salo}

\address[linz]{\LINZ}
\address[utu]{\UTU}

\begin{abstract}
It is well-known that the Toffoli gate and the negation gate together yield a universal gate set, in the sense that every permutation
of $\{0,1\}^n$ can be implemented as a composition of these gates. Since every bit operation that does not use all of the bits performs
an even permutation, we need to use at least one auxiliary bit to perform every permutation, and it is known that one bit is indeed enough. Without auxiliary bits, all even permutations can be implemented. We generalize these results to non-binary logic:
For any finite set $A$, a finite gate set can generate all even permutations of  $A^n$ for all $n$, without any auxiliary symbols.
This directly implies the previously published result that
a finite gate set can generate all permutations of $A^n$ when the cardinality of $A$ is odd, and that
one auxiliary symbol is necessary and sufficient to obtain all permutations when the cardinality of $A$ is even.
We also consider the conservative case,
that is, those permutations of $A^n$ that preserve the weight of the input word.
 The weight is the vector that records how many times
each symbol occurs in the word {\new or, more generally, the image of the word under a fixed monoid homomorphism from $A^*$
to a commutative monoid.}
 It turns out that no finite conservative gate set can, for all $n$,
implement all conservative even permutations
of  $A^n$ without auxiliary bits. But we provide a finite gate set that can implement all those conservative permutations
that are even within each weight class of $A^n$.
\end{abstract}

\begin{keyword}
reversible gates,
reversible circuits,
universal gates,
conservative gates,
reversible clones
\end{keyword}

\end{frontmatter}

\section{Introduction}

The study of reversible and conservative binary gates was pioneered in the 1970s and 1980s by
Toffoli and Fredkin \cite{FrTo82,toff80}.
Recently, Aaronson, Greier and Schaeffer \cite{aaronsonetal15} described all binary gate sets closed under
the use of auxiliary bits, as a prelude to their eventual goal of classifying these gate sets in the quantum case. It has been noted that ternary gates have similar, yet distinct properties \cite{yangetal05}.

In this article, we consider the problem of finitely-generatedness of various families of reversible logic gates without using auxiliary bits.
In the case of a binary alphabet, it is known that the whole set of {\new reversible} gates is not finitely generated in this strong sense, but the family of gates that perform an even permutation of $\{0,1\}^n$ is \cite{aaronsonetal15,musset97,xu15}. In \cite{yangetal05}, it is shown that for the ternary alphabet, the whole set of reversible gates is finitely generated.
In  \cite{lafont93} the result is announced for all odd alphabets, with a proof attributed to personal communication, which has recently been published as \cite{selinger16}. Another proof of this fact can be found in~\cite{boykett15}.
In this paper, we look at gate sets with arbitrary finite 
alphabets, and prove the natural generalization: the whole set of {\new reversible}  gates is finitely generated if and only if the alphabet is odd, and in the case of an even alphabet, the even permutations are finitely generated.

In \cite{xu15}, it is proved that in the binary case the conservative gates, reversible gates that preserve the numbers of symbols in the input (that is, its weight), are not finitely generated, even with the use of `borrowed bits', bits that may have any initial value but must return to their original value in the end. On the other hand, it is shown that with bits whose initial value is known (and suitably chosen), all permutations can be performed. We prove for all alphabets that the gates that perform an even permutation in every weight class are finitely generated, but the whole class of {\new conservative} permutations is far from being finitely generated (which implies in particular the result of \cite{xu15}).

We also consider more general 
conservation laws for gates. Assign to each letter of the alphabet as its weight an element of an arbitrary
commutative monoid. We say that a reversible gate conserves this assignment if the sum in the monoid of the weights of the
input symbols always equals the sum of the weights of the output symbols. This concept generalizes both the conservative gates and the unrestricted reversible gates. We prove that the generalized conservative gates that perform an even permutation in each weight class
are finitely generated. In contrast, the whole conservative class without the evenness requirement
is not finitely generated, provided there are sufficiently many non-trivial weight classes available.
This general result implies the special cases above.

Our methods are rather general, and the proofs in all cases follow the same structure. The negative aspect of these methods is that our universal gates are not the usual ones, and for example in the conservative case, one needs a bit of work (or computer time) to construct our universal gate family from the Fredkin gate.

We start by introducing our terminology, taking advantage of the
concepts of clone theory \cite{szendrei} applied to bijections
as developed in \cite{boykett15}, leading to what we call \emph{reversible clones} or \emph{revclones}, and \emph{reversible iterative algebras} or
\emph{\what{}s}.
We note in passing that one can also use category-theoretic terminology to discuss the same concepts, and this is the approach taken in \cite{lafont93,musset97}. In this terminology, what we call revitals are strict symmetric monoidal groupoids in the category where objects are sets of the form $A^n$ and the horizontal composition rule is given by Cartesian product. A formal difference is that unlike morphism composition in a category, our composition operation is total.

We generalize the idea of the Toffoli gate and Fredkin gate to what we call \emph{controlled permutations} and prove a general induction lemma showing that if we can add a single new control wire to a controlled permutation, we can add any number of control wires. 
We then show two combinatorial results about permutation groups that allow us
to simplify arguments about \what{}s.
This allows us to describe generating sets for various revclones and \what{}s of interest,
with the indication that these results will be useful for more general \what{} analysis,
as undertaken for instance in \cite{aaronsonetal15}.
While theoretical considerations show that finite generating sets do not
exist in some cases, in other cases explicit computational searches are able to provide
small generating sets.

\section{Background}
\label{sec:two}

Let $A$ be a finite set. We write $S_A$ or $\Sym(A)$ for the group of permutations or bijections of
$A$, $S_n$ for $Sym(\{1,\dots,n\})$ and
$\Alt(A)$ for the group of even permutations of $A$, $A_n=Alt(\{1,\dots,n\})$.
Let $B_n(A) = \{f:A^n\rightarrow A^n \;\vert\; f \mbox{ a bijection}\}=\Sym(A^n)$
be the group of $n$-ary bijections on $A^n$, and let
$B(A) = \cup_{n\in \N} B_n(A)$ be the collection of all bijections on powers of $A$.
We call them \emph{gates}.
{\new For $f\in B_n(A)$, we denote by $f_i:A^n\rightarrow A$ the $i$'th component of gate $f$, so that
$f(x_1,\dots ,x_n)= (f_1(x_1,\dots ,x_n),\dots,f_n(x_1,\dots ,x_n))$.}

We denote by
$\gen{X}$ the group generated by $X\subseteq B_n(A)$, a subgroup of $B_n(A)$.
{\new As $B_n(A)$ is a finite group,
$\gen{X}$ is also the monoid generated by $X$, that is, it consists of all compositions of functions in $X$.}
All our compositions are from right to left, so that $f\circ g: x\mapsto f(g(x))$. For consistency,
elements of permutation groups are then multiplied from right to left as well.

Each $\alpha \in S_n$ defines a \emph{wire permutation} $\pi_\alpha \in B_n(A)$
that permutes the coordinates of its input according to $\alpha$:
$$
\pi_\alpha(x_1,\dots,x_n) = (x_{\alpha^{-1}(1)},\dots, x_{\alpha^{-1}(n)}).
$$
The wire permutation $\id_n=\pi_{()}$ corresponding to the identity permutation $()\in S_n$ is the $n$-ary identity map.
Conjugating $f\in B_n(A)$ with a wire permutation $\pi_\alpha\in B_n(A)$ gives
$\pi_\alpha \circ f\circ \pi_\alpha^{-1}$, which we call a
\emph{rewiring} of $f$.
Rewirings of $f$ correspond to applying $f$ on arbitrarily ordered input wires,
{\new that is, the rewiring $g=\pi_\alpha \circ f\circ \pi_\alpha^{-1}$ satisfies
$(g_{\alpha(1)}(x),\dots , g_{\alpha(n)}(x)) = f(x_{\alpha(1)},\dots ,x_{\alpha(n)})$ for all $x=(x_1,\dots, x_n)$.
}

Any $f\in B_\ell(A)$ can be extended to $A^n$ for $n>\ell$ by applying it on any
selected $\ell$ coordinates while leaving
the other $n-\ell$ coordinates unchanged. Using the
clone theory derived  terminology in \cite{boykett15} we first define, for any $f\in B_n(A)$ and $g\in B_m(A)$, the
parallel application $f \oplus g \in B_{n+m}(A)$ by
\begin{align*}
(f \oplus g)(x_1,\dots x_{n+m}) = (&f_1(x_1,\dots,x_n), \dots, f_n(x_1,\dots,x_n), \\
& g_1(x_{n+1},\dots,x_{n+m}), \dots,g_m(x_{n+1},\dots,x_{n+m})).
\end{align*}
Then the \emph{extensions} of $f\in B_\ell(A)$ on $A^n$ are the rewirings of $f\oplus \id_{n-\ell}$.

Let $P\subseteq B(A)$. We denote by $\clone{P}\subseteq B(A)$ the set of gates that can be obtained
from the identity $\id_1$ and the elements of $P$ by compositions of gates of equal arity and by
extensions of gates of arities $\ell$ on $A^n$, for  $n\geq \ell$. Clearly $P\mapsto \clone{P}$
is a closure operator. Sets $C\subseteq B(A)$ such that $C=\clone{C}$ are called \emph{reversible iterative algebras}, or \emph{revitals}, for short.
We say that gate set $P$ \emph{generates} \what{} $C$ if $C=\clone{P}$. We say that  \what{} $C$ is \emph{finitely generated}
if there exists a finite set $P$ that generates it.

{\new
For example, the revital $\clone{\pi_{(1\; 2)}}$ generated by a single wire swap $\pi_{(1\; 2)}$ is exactly the set of all wire
permutations. Note that a revital is not required to contain wire permutations; revitals that contain all wire permutations are
known as \emph{reversible clones} or \emph{revclones} in clone theory~\cite{boykett15}.
}

 We sometimes refer to general elements of $B_n(A)$ as \emph{word permutations} to distinguish them from the wire permutations.
 In particular, by a wire swap we mean a function $f : A^2 \to A^2$ with $f(a,b) = (b,a)$ for all $a, b \in A$ (or an extension of such a function),
 while a word swap refers to a permutation $(u \; v) \in B_n(A)$ that swaps two individual words of the same length.
 Of course, a wire swap is a composition of word swaps, but the converse is not true. Similarly, and more generally,
 we talk about \emph{wire and word rotations}. A \emph{symbol permutation} is a permutation of $A$.

We are interested in finding out if some naturally arising \what{}s are finitely generated. First of all, we have the \emph{full \what{}}
$B(A)$ and the \emph {alternating \what{}} $\Cthree = \bigcup_n \Alt(A^n)$ that contains all even
permutations.

We also consider permutations that conserve the letters in their inputs.
For any $n\in \N$, define $w_n:A^n\rightarrow \N^A$, such that for all $x\in A^n$, $a\in A$,
$w_n(x)(a)$ the number of occurrences of $a$ in $x$. We say $w_n(u)$ is the \emph{weight} of the word $u$.
A mapping $f\in B_n(A)$ is \emph{conservative} if for all $x\in A^n$, $w_n(f(x)) = w_n(x)$, we let $Cons_n(A)\subseteq B_n(A)$ be the set of conservative maps of arity $n$.
Then $\Ctwo = \cup_{n\in \N}Cons_n(A)$ is the \emph{conservative \what}.
We also consider the set of conservative permutations
that perform an even permutation on each weight class, denoted by $\Cfour$, called the \emph{alternating conservative revital}.



We also investigate a more general concept of conservation. Let $M$ be a commutative monoid and let $\phi : A^* \to M$ be a monoid homomorphism. \emph{Weight classes} under $\phi$ consist of words of equal length having the same value under $\phi$.
Without loss of generality we may assume that $\phi$ also records the
length of the input word in the sense that $|x|\neq |y| \Rightarrow \phi(x)\neq\phi(y)$.
Indeed, if needed we can replace $\phi$ with the length-recording
homomorphism $x\mapsto (\phi(x),|x|)$ from $A^*$ to $M\times \N$ that has the same weight classes as $\phi$.

A mapping $f\in B_n(A)$ is \emph{$\phi$-conservative} if $\phi(f(x)) = \phi(x)$ holds for all $x\in A^n$.
Such $f$ performs a permutation on each weight class.
We let $Cons^\phi_n(A)\subseteq B_n(A)$ be the set of $\phi$-conservative maps of arity $n$.
Then $\Ctwophi = \cup_{n\in \N}Cons^\phi_n(A)$ is the \emph{$\phi$-conservative \what}.
We also consider $\phi$-conservative permutations
that perform an even permutation on each weight class.
We denote the set of such alternating $\phi$-conservative permutations by $\Cfourphi$, and call it the
\emph{alternating $\phi$-conservative revital}.

{\new

Note that $\Ctwo=\Ctwophi$ and $\Cfour=\Cfourphi$ for the homomorphism $\phi:A^*\longrightarrow \N^A$
that counts the numbers of occurrences of letters in a word, that is,
$\phi_{|A^n} = w_n$ for all $n\in \N$.
Also note that $B(A)=\Ctwophi$ and $\Cthree=\Cfourphi$ for the length
homomorphism $\phi: x\mapsto |x|$ from $A^*$ to $\N$. The weight classes of this homomorphism are namely
the full sets $A^n$. This also means that some of the
results concerning the full, the alternating, the conservative and the
alternating conservative revitals follow from our more general results for arbitrary $\phi$-conservative
revitals. However, we include the proofs for these special cases as gentle preludes to the general case.

}



Using the terminology in~\cite{xu15}, we say that gate
$f\oplus\id_k\in B_{n+k}(A)$ computes $f\in B_n(A)$ using $k$
\emph{borrowed} bits. The borrowed bits are auxiliary symbols in the computation of $f$
that can have arbitrary initial values, and at the end these values
must be restored unaltered. Regardless of the initial values of the borrowed bits, the permutation $f$ is computed on the
other $n$ inputs. We have cases where borrowed bits help
(Corollary~\ref{cor:borrowedcor}) and cases where they don't (Theorem~\ref{thm:XuGeneralization}).

\section{Induction Lemma}

In this section, we introduce the concept of controlled gate, a generalisation of the Toffoli and Fredkin gates.
With this definition, we are able to formulate a useful induction lemma.
This lemma formalizes  the following idea.
If we can build an $(n+1)$-ary controlled gate in a certain class from
gates of arity $n$, then by replacing each $n$-ary gate with its $(n+1)$-ary extension, we have a ``spare''
control line from each $n+1$ gate, which can then be attached to an extra control input to get an $(n+2)$-ary gate.


\begin{definition}
Let $k \in \N$ and $P \subseteq B_\ell(A)$. For $w \in A^k$ and $p \in P$, define the function
$\controlled{w}{p} : A^{k+\ell} \to A^{k+\ell}$ by
\[ \controlled{w}{p}(uv) = \left\{\begin{array}{ll}
uv & \mbox{if } u \neq w \\
u p(v) & \mbox{if } u = w
\end{array}\right. \]
where $u \in A^k$, $v \in A^\ell$.
The functions $\controlled{w}{p}$, and more generally their rewirings $\pi_\alpha\circ\controlled{w}{p}\circ\pi_\alpha^{-1}$ for
$\alpha\in S_{k+\ell}$,
are called \emph{$k$-controlled $P$-permutations},
and we denote this set of functions by $CP(k,P) \subseteq B_{k+\ell}(A)$.
We refer to $CP(P) = \bigcup_k CP(k,P)$ as \emph{controlled $P$-permutations}.
\end{definition}

When $P$ is a named family of permutations, such as the family of all swaps in $S_A$, we usually talk about `$k$-controlled swaps' instead of `controlled swap permutations'.
These are word swaps rather than wire swaps.
The Toffoli gate is a (particular) $2$-controlled symbol permutation, while the Fredkin gate is a (particular) $1$-controlled wire swap. Note that the `$k$' in `$k$-controlled' refers to the fact that the number of controlling wires is $k$. 
Of course, sometimes we want to talk about also the particular word $w$ in $\controlled{w}{p}(uv)$. To avoid ambiguity, we say such $\controlled{w}{p}(uv)$ is \emph{$w$-word controlled permutation}. In particular, the Toffoli gate is the $11$-word controlled symbol permutation, while the Fredkin gate is a $1$-word controlled wire swap.

The following lemma formalizes the idea of adding new common control wires to all gates
in a circuit.

\begin{lemma}
\label{lem:ExtraWireLemma}
Let $k,h,\ell \in \N$, $P \subseteq B_{\ell}(A)$ and $Q \subseteq B_n(A)$. If $CP(h,Q) \subseteq \clone{CP(k,P)}$, then
$CP(h+m,Q) \subseteq\clone{CP(k+m,P)}$ for all $m \in \N$.
\end{lemma}

\begin{proof}
Consider an arbitrary
$f \in CP(h+m,Q)$. Let $uv \in A^{h+m}$ be its control word where $u\in A^m$ and $v\in A^h$, and let $p\in Q$ be its permutation.
By the hypothesis, $f_{v,p}$ 
can be implemented by maps in $CP(k,P)$.
In all their control words, add the additional input $u$. This implements $f$ as a composition of maps in $CP(k+m,P)$, as required.
\qed
\end{proof}

The main importance of the lemma comes from the following corollary:

\begin{lemma}[Induction Lemma]
\label{lem:InductionLemma}
Let  $P \subseteq B_{\ell}(A)$ be such that
$CP(k+1,P) \subseteq \clone{CP(k,P)}$ for some
$k \in \N$. Then $\clone{CP(m,P)} \subseteq \clone{CP(n,P)}$ for all $m\geq n\geq k$.
\end{lemma}

\begin{proof}
We apply Lemma~\ref{lem:ExtraWireLemma}, setting $Q=P$ and $h=k+1$.
We obtain that
$CP(k+m+1,P) \subseteq \clone{CP(k+m,P)}$ for all $m \in \N$. As $\clone{\cdot}$ is a closure operator we have that
$\clone{CP(k+m+1,P)} \subseteq \clone{CP(k+m,P)}$ for all $m \in \N$. Hence
$$
\clone{CP(k,P)} \supseteq \clone{CP(k+1,P)} \supseteq \clone{CP(k+2,P)} \supseteq \dots
$$
which clearly implies the claimed result.
\qed
\end{proof}

By the previous lemma, in order to show that a \what{} $C$ is finitely generated, it is sufficient to find
some $P \subseteq B_\ell(A)$ such that
\begin{enumerate}
\item[(i)] $\gen{CP(m,P)} = C\cap B_{m+\ell}(A)$ for all large enough $m$, and
\item[(ii)] $CP(k+1,P) \subseteq \clone{CP(k,P)}$ for some $k$.
\end{enumerate}
Indeed, if $n\geq k$ is such that (i) holds for all $m\geq n$ then,
$$
C\cap B_{m+\ell}(A) = \gen{CP(m,P)} \subseteq \clone{CP(m,P)} \subseteq \clone{CP(n,P)},
$$
where the last inclusion follows from (ii) and the Induction lemma.
Note that by (i) we also have $CP(n,P)\subseteq C$.
So the finite subset $CP(n,P)$ of $C$ generates all but finitely many elements of $C$.

Condition (i) motivates the following definition.
\begin{definition}
\label{def:controluniversal}
Let $C$ be a \what{}. We say that a set of permutations $P \subseteq B_\ell(A)$ is \emph{$n$-control-universal}
for $C$ if $\gen{CP(n-\ell, P)} = C \cap B_n(A)$. More generally, a set $P \subseteq B(A)$ that may contain gates of
different arities, is $n$-control-universal for $C$ if
$$\gen{\bigcup_{\ell}\bigcup_{f\in B_\ell(A)\cap P} CP(n-\ell, P)} = C \cap B_n(A).$$
If $P$ is $n$-control-universal
for all large enough $n$, we say it is \emph{control-universal} for $C$.
\end{definition}

In the next two sections we find gate sets that are control-universal for \what{}s of interest.


\section{Some combinatorial group theory}

In this section, we prove some basic results  that the symmetric group is generated
by any `connected' family of swaps, and the alternating group by any `connected' family of $3$-cycles.
Similar results are folklore in combinatorial group theory, but we include full proofs for completeness' sake.

A \emph{hypergraph} is a set $V$ of vertices and a set $E$ of edges, $E \subseteq {\cal P}(V)$.
A $k$-hypergraph is a hypergraph where every edge has the same size, $k$.
A 2-hypergraph is a standard (undirected) graph.
A \emph{path} is a series of vertices $(v_1,\dots,v_n)$ such that for each pair $(v_i,v_{i+1})$ there is
an edge $e_i \in E$ such that $\{v_i,v_{i+1}\} \subseteq e_i$.
Two vertices $a,b\in V$ are \emph{connected} if there is a path $(v_1,\dots,v_n)$ with $v_1=a$ and $v_n=b$.
The relation of being connected is an equivalence relation and induces a partition of the vertices into \emph{connected components}.

If $H$ is a $3$-hypergraph, write $Graph(H)$ for the underlying graph of $H$:
$V(Graph(H)) = V(H)$ and $\{a,b\} \in E(Graph(H)) \iff \exists c: \{a,b,c\} \in E(H)$. Note that by our definition, the connected components of a $3$-hypergraph $H$ are precisely the connected components of $Graph(H)$.

\medskip

\noindent
{\bf Swap group}.
Let $H$ be a graph with nodes $V(H)$ and edges $E(H)$.
The \emph{swap group} $SG(H)$ is the group $G \leq \Sym(V(H))$
generated by swaps $(a \; b)$ with $\{a,b\} \in E(H)$.

\begin{lemma}
\label{lem:ConCompSym}
Let $H$ be a graph with connected components $H_1, \ldots, H_k$. Then
\[ SG(H) = \Sym(V(H_1)) \times \cdots \times \Sym(V(H_k)) \]
\end{lemma}

\begin{proof}
All of the swaps act in one of the components and there are no relations between them.
Thus, the swap group is the direct product of some permutation groups of the connected components.
We only need to show that in each connected component $H_i$, we can realize any permutation of vertices.
Since swaps generate the symmetric group, it is enough to show that if $a, b \in V(H_i)$ then
the swap $(a \; b)$ is in $SG(H)$. For this, let $a = a_1, a_2, \ldots, a_\ell = b$ be a
path from $a$ to $b$. Then
\[ (a,b) = (a_1 \; a_2) (a_2 \; a_3) \cdots  (a_{\ell-2} \; a_{\ell-1}) (a_{\ell} \; a_{\ell-1}) (a_{\ell-1} \; a_{\ell-2}) \cdots (a_3 \; a_2) (a_2 \; a_1). \]
\qed
\end{proof}


\medskip

\noindent
{\bf Cycling group}.
Let $H$ be a 3-hypergraph with nodes $V(H)$ and hyperedges $E(H)$.
The \emph{cycling group} $CG(H)$ of $H$
is the group $G \leq \Sym(V(H))$ generated by cycles $(a \; b \; c)$ where $\{a,b,c\} \in E(H)$. 

The following observation allows us to  take any element of the alternating group given
two 3-hyperedges that intersect in one or two places.

\begin{lemma}
$$
\begin{array}{rcl}
A_4 &=& \gen{(1 \; 2 \; 3), (2 \; 3 \; 4)},\\
A_5 &=& \gen{(1 \; 2 \; 3), (3 \; 4 \; 5)}.
\end{array}
$$
\end{lemma}

\begin{lemma}
\label{lem:ConCompAlt}
Let $H$ be a hypergraph, and let the connected components of $H$ be $H_1, \ldots, H_k$. Then
\[ CG(H) = \Alt(V(H_1)) \times \Alt(V(H_2)) \times \cdots \times \Alt(V(H_k)). \]
\end{lemma}

\begin{proof}
We prove the claim by induction on the number of hyperedges.
If there are no hyperedges, then  $CG(H) = \{\ID(V(H))\}$, as required.
Now, suppose that the claim holds for a hypergraph $H'$ and $H$ is obtained from $H'$ by adding a new
hyperedge $\{a, b, c\}$. If none of $a, b, c$ are part of a hyperedge of $H'$ or are fully contained
in a connected component of $Graph(H')$, then the claim is trivial, as either we add a new connected
component and by definition add its alternating group $\Alt_3 \cong \gen{(a \; b \; c)}$ to $CG(H)$,
or we do not modify the connected components at all.

Every permutation on the right side of the equality we want to prove decomposes into even permutations in
the components. We thus only have to show that a pair of swaps $(x \; y) (u \; v)$ can be implemented for any
vertices $x,y,u,v$ that are in the same component in $H$.
In components that do not intersect $\{a, b, c\}$, we can implement this permutation by
assumption.
If $x,y,u,v \in \{a,b,c\}$, the permutation is in $CG(H)$ by definition.
Since $(x \; y) (u \; v) = (x \; y) (a \; b)^2$ $(u \; v)$ it is enough to implement the permutation $(a \; b) (u \; v)$.

Now, we have two cases (up to reordering variables). Either $u \in \{a,b,c\}$ and
$v \notin \{a,b,c\}$ or $\{u,v\} \cap \{a,b,c\} = \emptyset$.
By analysing cases, the
claim  reduces to the $\Alt_5$ or the $\Alt_4$ situation of the previous Lemma.
\qed
\end{proof}

\section{Control-universality}
\label{sec:ControlUniversality}

Recall control-universality from Definition~\ref{def:controluniversal}: a set of permutations $P$ is control-universal for \what{} $C$
if, for all sufficiently large $n$, the $n$-ary gates in $C$ are precisely the compositions of $n$-ary controlled $P$-permutations.
As corollaries of the previous section, we will now find control-universal families of gates for our  \what{}s of interest:
the full \what{} $\Cone = \bigcup_n \sym(A^n)$, the conservative \what{} $\Ctwo$ and the more general $\phi$-conservative
\what{} $\Ctwophi$, the alternating \what{}
$\Cthree=\bigcup_n \alt(A^n)$ and the alternating conservative \what{} $\Cfour$ and its generalization alternating
$\phi$-conservative \what{} $\Cfourphi$. Corollaries~\ref{cor:P1},
~\ref{cor_consuniversal}, ~\ref{cor_generalizedconsuniversal}, ~\ref{cor:corollaryP3}, ~\ref{cor:Something} and
\ref{cor:AlternatingGeneralizedConservativeConnectedComponents} below provide control-universal gate sets
for these \what{}s. In each case the result is obtained by first constructing an appropriate graph (or $3$-hypergraph) whose edges (hyperedges) correspond to controlled swaps (controlled $3$-cycles) using $P$. With the help of Lemma~\ref{lem:ConCompSym}
(or Lemma~\ref{lem:ConCompAlt}, respectively) one can then conclude that all permutations (even permutations, respectively)
in appropriate sets can be generated.

Note that, as discussed at the end of this section,
the results concerning the generalized conservative gates, presented as items (b$^\phi$) and (d$^\phi$) below,
imply the results on the unrestricted and on the standard conservative gates, presented as items (a), (b), (c) and (d).
These weaker results are included as an introduction to our method.

\medskip

\noindent
{\bf (a) The full \what{} $\Cone$}.
Define the graph $G^{(1)}_{A,n}$ that has nodes $A^n$ and edges $\{u, v\}$ where the Hamming distance between $u$ and $v$ is one. The following is obvious:

\begin{lemma}
\label{lem:SwapGraph}
The graph $G^{(1)}_{A,n}$ is connected.
\end{lemma}

{\new
\noindent
Let  $P_1 = \{(a \; b) \;|\; a, b \in A\} \subseteq B_1(A)$, the set of symbol swaps. The swap group of
$G^{(1)}_{A,n}$ is then contained in $\gen{CP(n-1,P_1)}\subseteq \Sym(A^n)$.
By Lemmas~\ref{lem:ConCompSym} and \ref{lem:SwapGraph} the swap group is $\Sym(A^n)$, so $\gen{CP(n-1,P_1)}=\Sym(A^n)$.
We have the following:
}

\begin{corollary}
\label{cor:P1}
For all $n$, $P_1$ is $n$-control-universal for the \what{} $\Cone$.
\end{corollary}

\medskip

\noindent
{\bf (b) The conservative \what{} $\Ctwo$}.
Define the graph $G^{(2)}_{A,n}$ that has nodes $A^n$ and edges $\{uabv, ubav\}$ for all $a,b\in A$ and words $u,v$ with $|u|+|v|=n-2$.
Because any two words with the same weight are obtained from each other by permuting the positions of letters, and because swaps of adjacent letters generate such permutations, we have the following lemma.

\begin{lemma}
\label{lem:ConservativeConnectedComponents}
The connected components of $G^{(2)}_{A,n}$ are the weight classes.
\end{lemma}

{\new
\noindent
Let $P_2 = \{(ab \; ba) \;|\; a,b \in A\} \subseteq B_2(A)$.
For $n\geq 2$, the swap group of
$G^{(2)}_{A,n}$ is contained in $\gen{CP(n-2,P_2)} \subseteq \Ctwo\cap B_n(A)$. By
Lemmas~\ref{lem:ConCompSym} and \ref{lem:ConservativeConnectedComponents}
the swap group is $\Ctwo\cap B_n(A)$, so we have the following:
}

\begin{corollary}
\label{cor_consuniversal}
For all $n$,  $P_2$ is $n$-control-universal for the conservative \what{} $\Ctwo$.
\end{corollary}

\noindent
The classical Fredkin gate that operates on $\{0,1\}^3$ is a $1$-controlled $P_2$-permutation.
However, note that in the case of a larger alphabet the controlled $P_2$-permutations only swap a specific pair of symbols,
not just the arbitrary contents of two cells.


 \medskip

{\new
\noindent
{\bf (b$^\phi$) The $\phi$-conservative \what{} $\Ctwophi$}.
Let $\phi:A^*\to M$ be a homomorphism to a commutative monoid $M$, and let $m\in\N$. For $n\geq m$,
define the graph $G^{(2),\phi,m}_{A,n}$ that has nodes $A^n$ and edges $\{uxv, uyv\}$ for all $x,y\in A^m$ such that $\phi(x)=\phi(y)$,
and all words $u,v$ with $|u|+|v|=n-m$.

\begin{lemma}
\label{lem:GeneralizedConservativeConnectedComponents}
For any sufficiently large $m$, and any $n\geq m$, the connected components of $G^{(2),\phi,m}_{A,n}$ are the weight classes
of $A^n$ under $\phi$.
\end{lemma}

\begin{proof}
Each edge connects two vertices in the same weight class, so it only remains to show that the weight classes are connected.
Recall from Section~\ref{sec:two} that we may assume that $\phi$ separates words of different length:
$|x|\neq |y| \Rightarrow \phi(x)\neq\phi(y)$.

Because $\phi(A^*)$ is a finitely generated commutative monoid, it is finitely presented \cite[Chapter~5]{RoGa99}.  Let $m$ be the length of the longest word involved in a finite presentation of $\phi(A^*)$
under the generator set $\phi(A)$. Note that both sides of any relation
$x=y$ in the presentation are words of equal length $\leq m$ over alphabet $\phi(A)$.
 Let $n\geq m$, and let $s,t\in A^n$ be arbitrary vertices in the same weight class, that is, they satisfy $\phi(s)=\phi(t)$.
 It follows from the definition of monoid presentations that there is a sequence $s=w_1,w_2,\dots, w_k=t$
 of words $w_i\in A^n$ where consecutive words are connected by a relation of the monoid presentation, that is, for all $i$ there are
 words $u,v,x,y$ such that $w_i = uxv$ and $w_{i+1} = uyv$, and we have $\phi(x) = \phi(y)$ and $|x|=|y|\leq m$.
 We can assume $|x|=|y|=m$ by appending
 symbols from $u$ or $v$ to $x$ and $y$ if needed, so in fact all $\{w_i, w_{i+1}\}$ are edges in the graph.
\qed
\end{proof}

\noindent
Let $P^{\phi,m}_2 = \{(x \; y) \;|\; x,y\in A^m \mbox{ and } \phi(x)=\phi(y)\} \subseteq B_m(A)$. For $n\geq m$,
the swap group of $G^{(2),\phi,m}_{A,n}$ is contained in $\gen{CP(n-m,P^{\phi,m}_2)}\subseteq \Ctwophi\cap B_n(A)$.
By Lemmas~\ref{lem:ConCompSym} and \ref{lem:GeneralizedConservativeConnectedComponents}, when $m$ is large enough
the swap group is $\Ctwophi\cap B_n(A)$.

\begin{corollary}
\label{cor_generalizedconsuniversal}
 For sufficiently large $m$, set $P^{\phi,m}_2$ is $n$-control-universal for the $\phi$-conservative \what{} $\Ctwophi$,
for all $n\geq m$.
\end{corollary}

}

\medskip

\noindent
{\bf (c) The alternating \what{} $\Cthree$}.
Define the $3$-hypergraph $G^{(3)}_{A,n}$ that has nodes $A^n$ and hyperedges
$\{uabv, uacv, udbv\}$ where $a,b,c,d \in A$, $a \neq d$ and $b \neq c$, that is, all triples of words of
which two are at Hamming distance $2$ and others at distance~$1$, and the symbol
differences are in consecutive positions. When $n\geq 2$, any pair of words $u,v\in A^n$ of Hamming distance one are in a
common hyperedge, so by Lemma~\ref{lem:SwapGraph} we have the following:

\begin{lemma}
\label{lem:ConnectedAlt}
If $n \geq 2$, then $G^{(3)}_{A,n}$ is connected. 
\end{lemma}

\noindent
Let $P_3 = \{(ab \; ac \; db) \;|\; a,b,c,d \in A\} \subseteq B_2(A)$ so that
the cycling group of $G^{(3)}_{A,n}$ is contained in
$\gen{CP(n-2,P_3)}\subseteq \Alt(A^n)$. By Lemmas~\ref{lem:ConCompAlt} and \ref{lem:ConnectedAlt}
the cycling group contains all even permutations
of $A^n$ so we have the following:

\begin{corollary}
\label{cor:corollaryP3}
For all $n\geq 2$, $P_3$ is $n$-control-universal for the alternating \what{} $\Cthree$.
\end{corollary}

\medskip

{\new
\noindent
{\bf (d) The alternating conservative \what{} $\Cfour$}. Let $n\geq 3$.
Define the 3-hypergraph $G^{(4)}_{A,n}$ that has nodes $A^n$ and hyperedges $\{uxv, uyv, uzv\}$ for all distinct
$x,y,z\in A^3$ with equal weights, and all $u,v$ whose lengths satisfy $|u|+|v|=n-3$.
Then, for all distinct $a,b\in A$ and words $u,v\in A^*$ such that $|u|+|v|=n-2$,
the words $uabv$ and $ubav$ are in a common hyperedge. Indeed, any word of length three that contains
at least two different letters has at least three different letter permutations. We obtain from Lemma~\ref{lem:ConservativeConnectedComponents} the following:

\begin{lemma}
\label{lem:AltConsGraph}
If $n \geq 3$, then the connected components of $G^{(4)}_{A,n}$ are the weight classes.
\end{lemma}

Let $P_4 = \{(x \; y \; z) \;|\; x,y,z \in A^3 \mbox{ and } x,y,z \mbox{ have equal weights } \} \subseteq B_3(A)$
consist of all $3$-cycles inside weight classes in $A^3$.
The cycling group of $G^{(4)}_{A,n}$ is contained in
$\gen{CP(n-3,P_4)}\subseteq \Cfour\cap B_n(A)$. But by Lemmas~\ref{lem:ConCompAlt} and \ref{lem:AltConsGraph}
the cycling group is $\Cfour\cap B_n(A)$.
}

\begin{corollary}
\label{cor:Something}
Set $P_4$ is $n$-control-universal
for the alternating conservative \what{} $\Cfour$, for all $n\geq 3$.
\end{corollary}

{\new
\noindent
{\bf (d$^\phi$) The alternating $\phi$-conservative \what{} $\Cfourphi$}.
Let $\phi:A^*\to M$ be a homomorphism to a commutative monoid $M$, and let $m\in\N$. For $n\geq m$,
define the 3-hypergraph $G^{(4),\phi,m}_{A,n}$ that has nodes $A^n$ and hyperedges $\{uxv, uyv, uzv\}$ for all distinct
$x,y,z\in A^m$ such that $\phi(x)=\phi(y)=\phi(z)$,
and all words $u,v$ with $|u|+|v|=n-m$.

\begin{lemma}
\label{lem:AlternatingGeneralizedConservativeConnectedComponents}
For all sufficiently large $m$, and all  $n\geq m$, the connected components of $G^{(4),\phi,m}_{A,n}$ are the weight classes
of $A^n$ under $\phi$.
\end{lemma}

\begin{proof}
If $x,y\in A^{m-1}$ are such that  $\phi(x)=\phi(y)$ then for all $a\in A$ holds
$\phi(xa)=\phi(ya)=\phi(ay)=\phi(ax)$. If, moreover, $x\neq y$ then
there are at least three distinct elements among $ax$, $ay$, $xa$ and $ya$ and,
in addition, $ax\neq ay$ and $\ xa\neq ya$.
It follows that if $wxw', wyw' \in A^n$ where $n\geq m$ and
$x,y\in A^{m-1}$ are distinct and such that $\phi(x)=\phi(y)$, then there are three distinct words $s,t,u\in A^{m}$ such that
$\phi(s)=\phi(t)=\phi(u)$, and  $wxw', wyw' \in \{zsz', ztz', zuz'\}$ for some words $z,z'$. This means that for $n\geq m$
any vertices connected by an edge in the graph $G^{(2),\phi,m-1}_{A,n}$ belong to a hyperedge in the
3-hypergraph $G^{(4),\phi,m}_{A,n}$, and hence the connected components of the first one are subsets of the connected components of the latter one. By Lemma~\ref{lem:GeneralizedConservativeConnectedComponents}, if $m$ is sufficiently large the
connected components of $G^{(2),\phi,m-1}_{A,n}$ are the weight classes. Since the hyperedges of $G^{(4),\phi,m}_{A,n}$ only connect
vertices of the same weight class, we conclude that for sufficiently large $m$ and all $n\geq m$,
the connected components of the hypergraph $G^{(4),\phi,m}_{A,n}$ are precisely the weight classes.
\qed
\end{proof}

\noindent
Let $P^{\phi,m}_4 = \{(x \; y \; z) \;|\;
 x,y,z\in A^m \mbox{ and } \phi(x)=\phi(y)=\phi(z)\} \subseteq B_m(A)$.
The cycling group of $G^{(4),\phi,m}_{A,n}$ is contained in
$\gen{CP(n-m,P^{\phi,m}_4)}\subseteq \Cfourphi\cap B_n(A)$. By Lemmas~\ref{lem:ConCompAlt} and
\ref{lem:AlternatingGeneralizedConservativeConnectedComponents}
the cycling group is $\Cfourphi\cap B_n(A)$ when $m$ is sufficiently large.

\begin{corollary}
\label{cor:AlternatingGeneralizedConservativeConnectedComponents}
For sufficiently large $m$ and all $n\geq m$, set $P^{\phi,m}_4$ is $n$-control-universal for the alternating $\phi$-conservative \what{} $\Cfourphi$.
\end{corollary}
Note that Corollaries~\ref{cor:corollaryP3} and \ref{cor:Something} are special cases of
Corollary~\ref{cor:AlternatingGeneralizedConservativeConnectedComponents}, in the same way as
Corollaries~\ref{cor:P1} and \ref{cor_consuniversal} are special cases of
Corollary~\ref{cor_generalizedconsuniversal}. From the proof we can see that a bound on $m$ in
Corollary~\ref{cor_generalizedconsuniversal} is given by the maximum length of the words occurring
in the relations in the presentation of $\phi(A^*)$,
and the corresponding bound for $m$ in Corollary~\ref{cor:AlternatingGeneralizedConservativeConnectedComponents} is
one greater. If $\phi:A^*\rightarrow \N^A$ is the  homomorphism that counts the numbers of occurrences of letters,
the image $\phi(A^*)=\N^A$ has a presentation using words of length two (that defines the commutation of symbols);
hence the values $2$ and $3$ for $m$ in  Corollaries~\ref{cor_consuniversal} and \ref{cor:Something}, respectively.
If $\phi:A^*\rightarrow \N$ is the homomorphism $x\mapsto |x|$ then words of length one are used in the presentation stating that
$\phi(a)=\phi(b)$ for all $a,b\in A$; hence the values $1$ and $2$ for $m$ in  Corollaries~\ref{cor:P1} and \ref{cor:corollaryP3}, respectively.
}

\section{Finite generating sets of gates}

In order to apply the Induction Lemma we first observe that 2-controlled 3-word-cycles in any five element set can obtained from 1-controlled 3-word-cycles.

\begin{lemma}
\label{lem:fivelemma}
Let $X\subseteq A^n$ contain at least five elements, and let
\[ P=\{(x \; y \; z) \;|\; x,y,z \in X \mbox{ all distinct} \}\subseteq B_n(A) \]
consist of all $3$-word-cycles in $X$.
Then $CP(2,P) \subseteq \clone{CP(1,P)}$.
\end{lemma}

\begin{proof}
Let $x,y,z\in X$ be pairwise different, and
pick $s,t\in X$ so that $x,y,z,s,t$ are five distinct elements of $X$.
Let $p_1=(s \; t)(x \; y)$ and $p_2=(s \; t)(y \; z)$. Then $p_1$ and $p_2$ consist of two disjoint word swaps, so they are
both involutions. Moreover, $(x \; y \; z) = p_2 p_1 p_2 p_1$. Further, we have that
$$
\begin{array}{rcl}
p_1&=&(x \; s \; y)(s \; t \; x), \mbox{ and}\\
p_2&=&(y \; s \; z)(s \; t \; y).
\end{array}
$$
It is important to recall that we compose permutations from right to left.

Let $a,b\in A$ be arbitrary and consider the $2$-controlled $P$-permutation $f=\controlled{ab}{(x \; y \; z)} \in B_{2+n}(A)$ determined by the control word $ab$
and the 3-word-cycle $(x \; y \; z)$. Then $f=g \circ g$ where
$$
g = \controlled{*b}{p_2} \circ \controlled{a*}{p_1} =
\controlled{*b}{(y \; s \; z)}\circ
\controlled{*b}{(s \; t \; y)}\circ
\controlled{a*}{(x \; s \; y)}\circ
\controlled{a*}{(s \; t \; x)}
$$
is a composition of four 1-controlled $P$-permutations, where the star symbol indicates the control symbol not used by the gate.
See Figure~\ref{fig:threecycles} for an illustration. In the picture the inputs arrive from the left, so the gates
are composed in the reverse order compared to the text.

\begin{figure}[ht]
\begin{center}
\includegraphics[width=11.5cm]{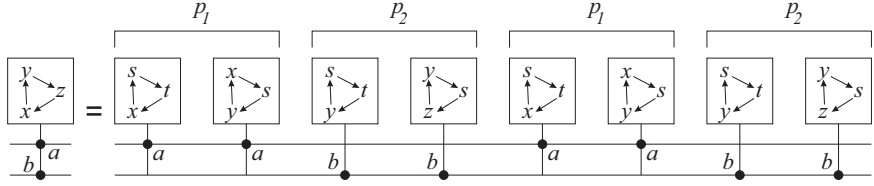}
\end{center}
\caption{A decomposition of the $ab$-controlled 3-word-cycle $(x \; y \; z)$ into a composition of eight 1-controlled 3-word-cycles. Note that in our illustrations the input-output direction of gates is from left to right. This means that the gates appear in the opposite order as in the main text where all compositions are written from right to left.}
\label{fig:threecycles}
\end{figure}

To verify that indeed $f=g \circ g$, consider an input $w=a'b'u$ where $a',b'\in A$ and $u\in A^n$. If $a'\neq a$ then
$g(w)=\controlled{*b}{p_2}(w)$, so that $g\circ g(w)=w=f(w)$ since $p_2$ is an involution. Analogously, if $b'\neq b$ then $g\circ g(w)=w=f(w)$, because $p_1$ is an involution.
Suppose then that $a'=a$ and $b'=b$. We have $g\circ g(w)=ab( p_2 p_1 p_2 p_1(u)) = f(w)$. We conclude that $f\in \clone{CP(1,P)}$, and because
$f$ was an arbitrary element of $CP(2,P)$, up to reordering the input and output symbols, the claim $CP(2,P)\subseteq \clone{CP(1,P)}$ follows.
\qed
\end{proof}

\begin{corollary}
\label{cor:threecycles}
Let $X\subseteq A^n, P\subseteq B_n(A)$ be as in Lemma~\ref{lem:fivelemma}. Then we have
$\clone{CP(m,P)} \subseteq \clone{CP(1,P)}$
for all $m\geq 1$.

\end{corollary}

\begin{proof}
Apply Lemma~\ref{lem:InductionLemma} with $k=1$.
\qed
\end{proof}

\subsection{The alternating and full \what{}s}

Assuming that $|A| > 1$, the set $X=A^3$ contains at least five elements.
For $P=\{(x \; y \; z) \;|\; x,y,z \in A^3\mbox{ all distinct}\}\subseteq B_3(A)$ we then have, by Corollary~\ref{cor:threecycles}, that
$\clone{CP(m,P)} \subseteq \clone{CP(1,P)}$ for all $m\geq 1$.

Recall that $P_3=\{(ab \; ac \; db) \;|\; a,b,c,d \in A\} \subseteq B_2(A)$ is $n$-control-universal for the alternating \what{} $\Cthree$, for $n\geq 2$ (Corollary~\ref{cor:corollaryP3}).
Clearly $\cp{1}{P_3}\subseteq P \subseteq \clone{\cp{0}{P}}$, so by Lemma~\ref{lem:ExtraWireLemma}, for any $m\geq 1$,
$$
\cp{m+1}{P_3} \subseteq \clone{\cp{m}{P}} \subseteq \clone{\cp{1}{P}}.
$$
Hence $\Cthree\cap B_{m+3}(A) = \gen{CP(m+1,P_3)} \subseteq \clone{\cp{1}{P}}$.
We conclude that $\clone{\cp{1}{P}}$ contains all permutations of $\Cthree$ except the ones in $B_1(A), B_2(A)$ and $B_3(A)$. We have proved the following theorem.

\begin{theorem}
\label{thm:AltFiniteGen}
The alternating \what{} $\Cthree$ is finitely generated. Even permutations of $A^4$ generate all even permutations of $A^n$ for all $n\geq 4$.
\end{theorem}

\begin{corollary}[\cite{selinger16,boykett15}]
\label{cor:FullFiniteGen}
Let $|A|$ be odd. Then the full \what{} $\Cone$ is finitely generated.
The permutations of $A^4$ generate all permutations of $A^n$ for all $n\geq 4$.
\end{corollary}

\begin{proof}
Let $|A|>1$ be odd.
Let $P$ be the set of all permutations of $A^4$, and let $n\geq 4$. By Theorem~\ref{thm:AltFiniteGen},
the closure $\clone{P}$ contains all even permutations of $A^n$. The set
$P$ also contains an odd permutation $f$, say the word swap $(0000\; 1000)$.
Consider
$\pi=f\oplus \id_{n-4}\in B_n(A)$  that applies the swap $f$ on the first
four input symbols and keeps the others unchanged.
This $\pi$ is an odd permutation because it consists of $|A|^{m-4}$ disjoint swaps and $|A|$ is odd.
Because $\clone{P}\cap B_n(A)$ contains all even permutations of $A^n$ and an odd one, it contains all permutations.
\qed
\end{proof}

Recall that if
a circuit implements the permutation $f\oplus\id_k\in B_{n+k}(A)$,
we say it implements $f \in B_n(A)$ using $k$
borrowed bits.

\begin{corollary}
\label{cor:borrowedcor}
The \what{} $\Cone$ is finitely generated using at most one borrowed bit.
\end{corollary}

\begin{proof}
For $|A|$ odd the claim follows from Corollary~\ref{cor:FullFiniteGen}.
When $A$ is even then the permutations $f\oplus\id$ with one borrowed bit are all even, so the claim follows
from Theorem~\ref{thm:AltFiniteGen}.
\qed
\end{proof}

\subsection{The alternating conservative \what{}}
\label{sec:altcons}

Every non-trivial weight class of $A^5$ contains at least five elements.
(The trivial weight-classes are the singletons $\{a^5\}$ for $a\in A$.) For every non-trivial weight class $X$ we set
$P_X=\{(x \; y \; z) \;|\; x,y,z \in X\}\subseteq B_5(A)$ for the 3-word-cycles in $X$. By Corollary
\ref{cor:threecycles} we know that
$\clone{CP(m,P_X)} \subseteq \clone{CP(1,P_X)}$ for all $m\geq 1$. Let $P$ be the union of $P_X$ over all non-trivial weight classes $X$.
Then, because $\clone{\cdot}$ is a closure operator,  also $\clone{CP(m,P)} \subseteq \clone{CP(1,P)}$ for all $m\geq 1$.

By Corollary~\ref{cor:Something}, the set  $P_4\subseteq B_3(A)$
of $3$-cycles within weight classes
is $n$-control-universal for the alternating conservative
\what{} $\Cfour$, for all $n\geq 3$.
Let $m\geq 1$.
Because $\cp{2}{P_4}\subseteq P \subseteq \clone{\cp{0}{P}}$, by Lemma~\ref{lem:ExtraWireLemma} we have
$$
\cp{m+2}{P_4} \subseteq \clone{\cp{m}{P}} \subseteq \clone{\cp{1}{P}}.
$$
Hence $\Cfour\cap B_{m+5}(A) = \gen{CP(m+2,P_4)} \subseteq \clone{\cp{1}{P}}$.
We conclude that $\clone{\cp{1}{P}}$ contains all permutations of $\Cfour$
except possibly the ones in
$B_\ell(A)$ for $\ell\leq 5$.

\begin{theorem}
The alternating conservative \what{} $\Cfour$ is finitely generated. A gate set generates the whole $\Cfour$ if it generates,
for all $n\leq 6$,
the conservative permutations of $A^n$ that are even on all weight classes.
\qed
\end{theorem}

{\new

\subsection{The generalized alternating conservative case}

Let $\phi:A^*\rightarrow M$ be a homomorphism to a commutative monoid $M$. Analogously to the standard alternating conservative case,
we have the following:

\begin{theorem}
\label{thm:altphicons}
The alternating $\phi$-conservative \what{} $\Cfourphi$ is finitely generated.
\end{theorem}

\begin{proof}
Fix $m$ to be sufficiently large so that the conclusions of
Lemma~\ref{lem:AlternatingGeneralizedConservativeConnectedComponents} and
Corollary~\ref{cor:AlternatingGeneralizedConservativeConnectedComponents} hold, and let $n=\max(5, m+1)$.

Consider a non-trivial weight class $X\subseteq A^n$ under $\phi$. Let us prove
that $X$ contains at least five elements.
By Lemma~\ref{lem:AlternatingGeneralizedConservativeConnectedComponents} weight class
$X$ contains the vertices of some hyperedge of $G^{(4),\phi,m}_{A,n}$.
These hyperedges have form $\{uxv, uyv, uzv\}$ for distinct
$x,y,z\in A^m$. Because $n\geq m+1$, either $u$ or $v$ is a non-empty word. It follows that
$X$ contains a word with at least two distinct letters. Any letter permutation of such a word is also in
$X$, and because $n\geq 5$ there are at least five distinct such permutations, proving $|X|\geq 5$.

Next we apply Corollary~\ref{cor:threecycles} as in Section~\ref{sec:altcons} above.
Let
\[ P=\{(x \; y \; z) \;|\;
 x,y,z\in A^n \mbox{ and } \phi(x)=\phi(y)=\phi(z)\} \subseteq B_n(A) \]
be the set of $3$-word
 cycles inside the weight classes of $A^n$. Because all non-trivial weight classes
contain at least five elements we get from Corollary~\ref{cor:threecycles}
that $\clone{CP(k,P)} \subseteq \clone{CP(1,P)}$ for all $k\geq 1$.

Consider $P^{\phi,m}_4\subseteq B_m(A)$, the set of $3$-word cycles within 
the weight classes of $A^m$.
 Because $CP(n-m,P^{\phi,m}_4)\subseteq P \subseteq \clone{CP(0,P)}$,
we get from Lemma~\ref{lem:ExtraWireLemma} that for all $k\geq 1$
$$
\cp{n-m+k}{P^{\phi,m}_4} \subseteq \clone{\cp{k}{P}} \subseteq \clone{\cp{1}{P}}.
$$

By Corollary~\ref{cor:AlternatingGeneralizedConservativeConnectedComponents} the set $P^{\phi,m}_4$ is
$n+k$-control-universal for $\Cfourphi$.
Hence $\Cfourphi\cap B_{n+k}(A) = \gen{\cp{n-m+k}{P^{\phi,m}_4}} \subseteq \clone{\cp{1}{P}}$.
We conclude that $\clone{\cp{1}{P}}$ contains all permutations of $\Cfourphi$
except possibly the ones in
$B_\ell(A)$ for $\ell \leq n$.

\qed
\end{proof}
}

\section{Non-finitely generated \what{}s}

It is well known that the full \what{} is not finitely generated over even alphabets. The reason is that any permutation
$f\in B_n(A)$ can only compute even permutations on $A^{m}$ for $m>n$.

\begin{theorem}[\cite{toff80}]
For even $|A|$, the full \what{} $\Cone$ is not finitely generated.
\end{theorem}

By another parity argument we can also show that the conservative \what{} $\Ctwo$ is not finitely generated on any
non-trivial alphabet, not even if infinitely many borrowed bits are available.
This generalizes a result in~\cite{xu15} on binary alphabets. Our proof is based on the same parity sequences as the one
in~\cite{xu15}, where these sequences are computed concretely for generalized Fredkin gates. However, our observation only relies on the (necessarily) low rank of a finitely-generated group of such parity sequences, and the particular conserved quantity is not as important.


Let $n\in\N$, and let $\W$ be the family of the weight classes of $A^n$.
For any  $f\in \Ctwo\cap B_n(A)$ and any weight class $X\in \W$,
the restriction $f|_{X}$ of $f$ on the weight class $X$ is a permutation of $X$.
Let $\psi(f)_X\in \Z_2$ be its parity. Clearly, $\psi(f\circ g)_X=\psi(f)_X+\psi(g)_X$ modulo two,
so $\psi$ defines a group homomorphism from $\Ctwo\cap B_n(A)$ to the additive abelian group $(\Z_2)^{\W}$.
The image $\psi(f)$ that records all $\psi(f)_X$ for all $X\in \W$ is the \emph{parity sequence} of $f$.
Because each element of the commutative group $(\Z_2)^{\W}$ is an involution, it follows that the subgroup generated by
any $k$ elements has cardinality at most $2^k$.

Consider then a function $f\in \Ctwo\cap B_{\ell}(A)$ for $\ell\leq n$.
Its application $f_n = f\oplus\id_{n-\ell} \in B_n(A)$
on length $n$ inputs is conservative, so it has the associated parity sequence $\psi(f')$, which we
denote by $\psi_n(f)$. Note that any conjugate $\pi_\alpha\circ f\circ \pi_\alpha^{-1}$ of $f$ by a wire permutation $\pi_\alpha$ has the same parity sequence, so the parity sequence
does not depend on which input wires we apply $f$ on.

Let $f^{(1)}, f^{(2)}, \dots , f^{(m)}\in \Ctwo$  be a finite generator set, and
let us denote by $C\subseteq \Ctwo$ the \what{} they generate.
Let $n \geq 2$ be larger than the arity of any $f^{(i)}$. Then $C \cap B_n(A)$
is the group generated by the applications $f^{(1)}_n, f^{(2)}_n, \dots , f^{(m)}_n$
of the generators on length $n$ inputs, up to conjugation by wire permutations.
We conclude that there are at most $2^m$ different parity sequences on $C\cap B_n(A)$, for all sufficiently large $n$. We have proved the following lemma.


\begin{lemma}
\label{lem:paritysequencelemma}
Let $C$ be a finitely generated sub\what{} of $\Ctwo$. Then there exists a constant $N$ such that, for all $n$,
the elements of $C\cap B_n(A)$ have at most $N$ different parity sequences.
\end{lemma}

Now we can prove the following negative result. Not only does it state
that no finite gate set generates the conservative \what{}, but even that there
necessarily remain conservative permutations that cannot be obtained
using any number of borrowed bits.

\begin{theorem}
\label{thm:XuGeneralization}
Let $|A|>1$. The conservative \what{} $Cons(A)$ is not finitely generated. In fact, if
$C\subseteq Cons(A)$ is finitely generated then there exists $f\in
Cons(A)$ such that $f\oplus \id_k\not\in C$ for all $k=0,1,2,\dots$.
\end{theorem}

\begin{proof}
Let $0,1 \in A$ be distinct.
Let $C$ be a finitely generated sub\what{} of $\Ctwo$, and
let $N$ be the constant from
Lemma~\ref{lem:paritysequencelemma} for $C$. Let us fix $n\geq
N+2$. For each $i=1,2,\dots, N+1$,
consider the non-trivial weight classes $X_i$ containing the words of $A^n$ with $i$ letters 1 and $n-i$ letters 0.
For each $i$, let $f_i\in \Ctwo\cap B_n(A)$ be a permutation that swaps two elements of $X_i$,
keeping all other elements of $A^n$ unchanged. This $f_i$ is odd on $X_i$ and even on all other weight classes, so all $f_i$ have different parity sequences.
We conclude that
some $f_i$ is not in $C$.

For the second, stronger claim, we continue by considering
an arbitrary $k\in\N$. For $i=1,2,\dots , N+1$, let $X^{(k)}_i$ be the parity class
of $A^{n+k}$ containing the words with $i$ letters 1 and $n+k-i$ letters 0.
Note that  $f^{(k)}_i=f_i\oplus \id_k$ is a single word swap on $X^{(k)}_i$ and the identity map 
on all $X^{(k)}_j$ with $j<i$. This means that
the parity sequences of $f^{(k)}_1, f^{(k)}_2, \dots , f^{(k)}_{N+1}$ are all different, hence some $f^{(k)}_i$ is not in $C$.
But then, for some $i\in\{1,2,\dots, N+1\}$, there are infinitely many
$k\in\N$ with the property that $f^{(k)}_i=f_i\oplus \id_k$ is not in
$C$. This means that $f_i\oplus \id_k\not\in C$ for \emph{any} $k\in \N$
as $f_i\oplus \id_k\in C$ implies that $f_i\oplus \id_\ell\in C$ for
all $\ell>k$. The permutation $f=f_i$ has the claimed property.
\qed
\end{proof}

{\new

Let us consider next the more general case of the $\phi$-conservative \what{} $\Ctwophi$. The following example shows that
the second, stronger statement in Theorem~\ref{thm:XuGeneralization} concerning borrowed bits
does not hold for all choices of $\phi$.

\bigskip

\noindent
{\bf Example.}
Let $A=\{0,1,2,3\}$ and consider the morphism $\phi:A^*\longrightarrow \N^2$
defined by $\phi(0)=\phi(2)=(1,0)$ and $\phi(1)=\phi(3)=(0,1)$. Now $\phi(x)=\phi(y)$ if and only if
$x$ and $y$ contain equally many even letters, and equally many odd letters. Let $C=\Cfourphi$, a sub-revital
of $\Ctwophi$ that is  finitely generated by Theorem~\ref{thm:altphicons}.

Let $f\in \Ctwophi$ be arbitrary. Then $f\oplus \id$ is even on every $\phi$-weight class:
$x0$ and $x2$ are in the same weight class, as are $x1$ and $x3$. In each case $f\oplus \id$
only changes the word $x$ independent of the last symbol, so $f\oplus \id$ performs in each
weight class two copies of the same permutation. We conclude that $f\oplus \id\in C$.
This means that $\Ctwophi$ is finitely generated using one borrowed bit.
\xqed{$\triangle$}

\bigskip

However, we can prove that $\Ctwophi$ itself -- without borrowed bits -- is not finitely
generated if there are enough non-trivial weight classes.
Let us call $\phi:A^*\rightarrow M$ \emph{infinite-dimensional} if for all $m\in\N$ there exists $n\in\N$ such that
there are at least $m$ weight classes with more than one word in $A^n$. The homomorphism in the previous example is infinite dimensional. In contrast, if $M$ is finite then $\phi$ can not be infinite-dimensional.

\begin{theorem}
\label{thm:ConservedQuantityNonFG}
Let homomorphism $\phi:A^*\rightarrow M$ be infinite-dimensional.
Then revital $\Ctwophi$ is not finitely generated.
\end{theorem}

The theorem shows, for example, that the \what{} of functions in $B(\{0,1,2\})$ that preserve the number of zeroes, and preserve the number of ones modulo $k$, is not finitely generated.

\begin{proof}
Suppose $\Ctwophi$ is finitely generated.
We define the parity sequence of $f\in \Ctwophi\cap B_n(A)$ analogously to the case of $\Ctwo$:
it records the parities of $f$ restricted to
$\phi$-weight classes of $A^n$.
Lemma~\ref{lem:paritysequencelemma} holds with an analogous proof:
there exists a constant $N$ such that, for all $n$, the elements of $\Ctwophi\cap B_n(A)$ have at most $N$ different parity sequences.
Because $\phi$ is infinite-dimensional, for some $n$ there are at least $N+1$ non-trivial
weight classes $X_1, X_2,\dots ,X_{N+1}$ in $A^n$. For $i\in\{1,2,\dots, N+1\}$, let $f_i\in B_n(A)$ swap two elements
of class $X_i$ and keep all other words unchanged. This $f_i$ is in $\Ctwophi$,
and it is odd on $X_i$ and even on all other weight classes. Thus there are $N+1$ functions $f_1,f_2,\dots , f_{N+1}$
in $\Ctwophi\cap B_n(A)$ with different
parity sequences, a contradiction.
\qed

\end{proof}

Note that the theorem is not valid without the infinite-dimensionality condition: For example the full revital
$\Cone$ is $\Ctwophi$ for the length homomorphism $\phi:x\mapsto |x|$, but it is finitely generated when $|A|$ is odd (Corollary~\ref{cor:FullFiniteGen}).

}

\section{Concrete generating families}
\label{secsearches}

We have found finite generating sets for \what{}s in the general and the conservative cases. Our generating sets are of the form `all controlled $3$-word cycles that are in the family', and the reader may wonder whether there are more natural gate families that generate these classes. Of course, by our results, there is an algorithm for checking whether a particular set of gates is a set of generators, and in this section we give some examples.

First, we observe that $CP(2,P_1)$ (that is, $2$-controlled symbol swaps) generate all permutations of $A^3$ and all even permutations of $A^n$ for all $n\geq 4$. Indeed, by Corollary~\ref{cor:P1} they generate $B_3(A)$, and by Figure~\ref{fig:tof} they generate $CP(2,P_3)$ (the 2-controlled $3$-cycles of length-two words). These in turn, by Corollary~\ref{cor:corollaryP3}, generate all even permutations of $A^4$
which is enough by Theorem~\ref{thm:AltFiniteGen} to get all even permutations on $A^n$ for $n\geq 4$.

It is easy to see that $CP(2,P_1)$ in turn is generated by all symbol swaps and the $w$-word-controlled symbol swaps for a single $w \in A^2$. In particular in the case of binary alphabets, we obtain that the alternating \what{} is generated by the Toffoli gate and the negation gate, which was also proved in \cite{xu15}.

\begin{figure}[ht]
\begin{center}
\includegraphics[width=11.5cm]{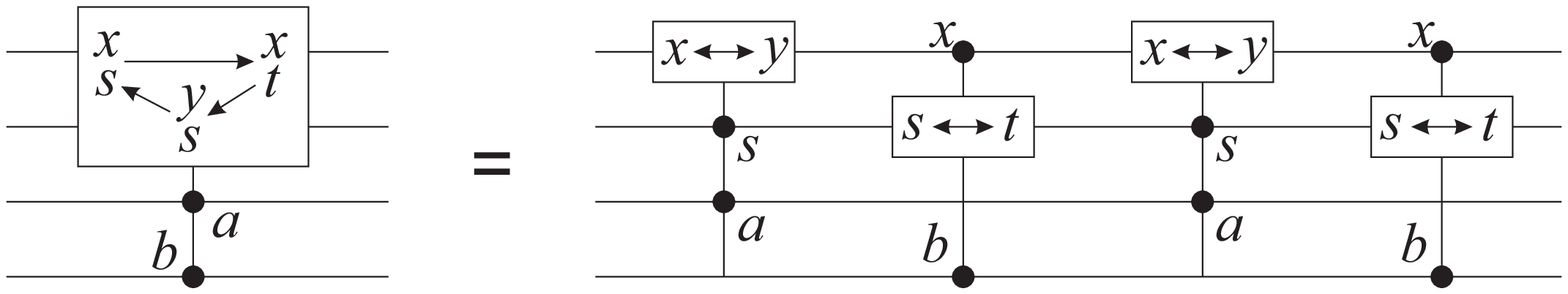}
\end{center}
\caption{A decomposition of the $ab$-controlled 3-cycle $(xs \; xt \; ys)$ into a composition of four 2-controlled swaps.}
\label{fig:tof}
\end{figure}

In the conservative binary case, the Fredkin gate is known to be universal (in the sense of auxiliary bits, see \cite{xu15}). The Fredkin gate is, due to the binary alphabet, both the unique 1-word-controlled wire swap and the unique nontrivial conservative 1-word-controlled word swap. The natural generalizations would be to show that in general the $1$-controlled wire swaps or conservative word swaps generate the alternating conservative \what{}. We do not prove this, but do show how the universality of the Fredkin gate follows from our results and a bit of computer search.

The following shows that the 00-word-controlled rotation and the $01$-word-controlled are generated by the 0-word-controlled rotation. These implementations were found by computer search, and both are optimal in the number of gates.

\begin{lemma}
\label{lem:FirstLemmaOfThisSection}
The $00$-word-controlled (resp. $01$-word-controlled) three-wire rotation can be implemented with nine (resp. eight) $0$-word-controlled three-wire rotations.
\end{lemma}

\begin{proof}
See Figure~\ref{fig:ccrot} for the diagrams of these implementations. We give the implementation of the $00$-word-controlled rotation also in symbols: Let $A = \{0,1\}$ and $R \in B_3(A)$ be the rotation $R = \pi_{(1\,2\,3)}$. Write $\rho_{a,b,c,d}(f)$ for $f$ applied to cells $a,b,c,d$ in that order.
\begin{align*}
\controlled{00}{R} = \;
&\rho_{1,0,2,3}(\controlled{0}{R}) \circ
\rho_{3,1,4,2}(\controlled{0}{R}) \circ
\rho_{1,0,2,4}(\controlled{0}{R}) \circ \\
&\rho_{3,0,1,2}(\controlled{0}{R}) \circ
\rho_{0,1,3,4}(\controlled{0}{R}) \circ
\rho_{1,2,3,4}(\controlled{0}{R}) \circ \\
&\rho_{0,1,4,3}(\controlled{0}{R}) \circ
\rho_{1,0,2,3}(\controlled{0}{R}) \circ
\rho_{3,0,2,4}(\controlled{0}{R})
\end{align*}
\qed
\end{proof}

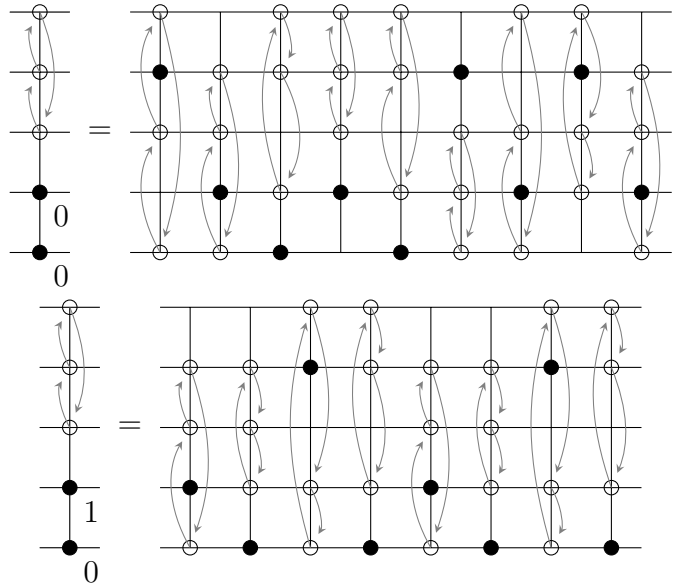
\begin{figure}[h]
\begin{center}
 \begin{tikzpicture}[scale = 0.8] 

\draw[fill] (-2.5,0) grid (-1.5,4);
\draw[fill] (-2,0) circle (0.12);
\draw[fill] (-2,1) circle (0.12);
\draw[] (-2,2) circle (0.12);
\draw[] (-2,3) circle (0.12);
\draw[] (-2,4) circle (0.12);
\draw[-stealth,shorten >=6pt,color=gray] (-2,2) to[bend left=30] (-2,3);
\draw[-stealth,shorten >=6pt,color=gray] (-2,3) to[bend left=30] (-2,4);
\draw[-stealth,shorten >=6pt,color=gray] (-2,4) to[bend left=22] (-2,2);
\node[below right=0.05] () at (-2,0) {$0$}; \node[below right=0.05] () at (-2,1) {$0$};

\node () at (-1,2) {$=$};

\draw (-0.5,0) grid (8.5,4);

\draw[fill] (0,3) circle (0.12); \draw[fill] (1,1) circle (0.12); \draw[fill] (2,0) circle (0.12);
\draw[fill] (3,1) circle (0.12); \draw[fill] (4,0) circle (0.12); \draw[fill] (5,3) circle (0.12);
\draw[fill] (6,1) circle (0.12); \draw[fill] (7,3) circle (0.12); \draw[fill] (8,1) circle (0.12);

\draw[] (0,0) circle (0.12); \draw[] (1,0) circle (0.12); \draw[] (2,1) circle (0.12);
\draw[] (3,2) circle (0.12); \draw[] (4,1) circle (0.12); \draw[] (5,0) circle (0.12);
\draw[] (6,0) circle (0.12); \draw[] (7,1) circle (0.12); \draw[] (8,0) circle (0.12);

\draw[] (0,2) circle (0.12); \draw[] (1,2) circle (0.12); \draw[] (2,3) circle (0.12);
\draw[] (3,3) circle (0.12); \draw[] (4,3) circle (0.12); \draw[] (5,1) circle (0.12);
\draw[] (6,2) circle (0.12); \draw[] (7,2) circle (0.12); \draw[] (8,2) circle (0.12);

\draw[] (0,4) circle (0.12); \draw[] (1,3) circle (0.12); \draw[] (2,4) circle (0.12);
\draw[] (3,4) circle (0.12); \draw[] (4,4) circle (0.12); \draw[] (5,2) circle (0.12);
\draw[] (6,4) circle (0.12); \draw[] (7,4) circle (0.12); \draw[] (8,3) circle (0.12);


\draw[-stealth,shorten >=6pt,color=gray] (0,0) to[bend left=30] (0,2);
\draw[-stealth,shorten >=6pt,color=gray] (0,2) to[bend left=30] (0,4);
\draw[-stealth,shorten >=6pt,color=gray] (0,4) to[bend left=16] (0,0);

\draw[-stealth,shorten >=6pt,color=gray] (1,0) to[bend left=30] (1,2);
\draw[-stealth,shorten >=6pt,color=gray] (1,2) to[bend left=30] (1,3);
\draw[-stealth,shorten >=6pt,color=gray] (1,3) to[bend left=19] (1,0);

\draw[-stealth,shorten >=6pt,color=gray] (2,4) to[bend left=30] (2,3);
\draw[-stealth,shorten >=6pt,color=gray] (2,3) to[bend left=30] (2,1);
\draw[-stealth,shorten >=6pt,color=gray] (2,1) to[bend left=19] (2,4);

\draw[-stealth,shorten >=6pt,color=gray] (3,2) to[bend left=30] (3,3);
\draw[-stealth,shorten >=6pt,color=gray] (3,3) to[bend left=30] (3,4);
\draw[-stealth,shorten >=6pt,color=gray] (3,4) to[bend left=22] (3,2);

\draw[-stealth,shorten >=6pt,color=gray] (4,1) to[bend left=30] (4,3);
\draw[-stealth,shorten >=6pt,color=gray] (4,3) to[bend left=30] (4,4);
\draw[-stealth,shorten >=6pt,color=gray] (4,4) to[bend left=19] (4,1);

\draw[-stealth,shorten >=6pt,color=gray] (5,0) to[bend left=30] (5,1);
\draw[-stealth,shorten >=6pt,color=gray] (5,1) to[bend left=30] (5,2);
\draw[-stealth,shorten >=6pt,color=gray] (5,2) to[bend left=22] (5,0);

\draw[-stealth,shorten >=6pt,color=gray] (6,0) to[bend left=30] (6,2);
\draw[-stealth,shorten >=6pt,color=gray] (6,2) to[bend left=30] (6,4);
\draw[-stealth,shorten >=6pt,color=gray] (6,4) to[bend left=16] (6,0);

\draw[-stealth,shorten >=6pt,color=gray] (7,4) to[bend left=30] (7,2);
\draw[-stealth,shorten >=6pt,color=gray] (7,2) to[bend left=30] (7,1);
\draw[-stealth,shorten >=6pt,color=gray] (7,1) to[bend left=19] (7,4);

\draw[-stealth,shorten >=6pt,color=gray] (8,0) to[bend left=30] (8,2);
\draw[-stealth,shorten >=6pt,color=gray] (8,2) to[bend left=30] (8,3);
\draw[-stealth,shorten >=6pt,color=gray] (8,3) to[bend left=19] (8,0);
\end{tikzpicture}
%
%
%
%
%
\begin{tikzpicture}[scale = 0.8] 

\draw[fill] (-2.5,0) grid (-1.5,4);
\draw[fill] (-2,0) circle (0.12);
\draw[fill] (-2,1) circle (0.12);
\draw[] (-2,2) circle (0.12);
\draw[] (-2,3) circle (0.12);
\draw[] (-2,4) circle (0.12);
\draw[-stealth,shorten >=6pt,color=gray] (-2,2) to[bend left=30] (-2,3);
\draw[-stealth,shorten >=6pt,color=gray] (-2,3) to[bend left=30] (-2,4);
\draw[-stealth,shorten >=6pt,color=gray] (-2,4) to[bend left=22] (-2,2);
\node[below right=0.05] () at (-2,0) {$0$}; \node[below right=0.05] () at (-2,1) {$1$};

\node () at (-1,2) {$=$};

\draw (-0.5,0) grid (7.5,4);

\draw[fill] (0,1) circle (0.12); \draw[fill] (1,0) circle (0.12); \draw[fill] (2,3) circle (0.12);
\draw[fill] (3,0) circle (0.12); \draw[fill] (4,1) circle (0.12); \draw[fill] (5,0) circle (0.12);
\draw[fill] (6,3) circle (0.12); \draw[fill] (7,0) circle (0.12);

\draw[] (0,0) circle (0.12); \draw[] (1,1) circle (0.12); \draw[] (2,0) circle (0.12);
\draw[] (3,1) circle (0.12); \draw[] (4,0) circle (0.12); \draw[] (5,1) circle (0.12);
\draw[] (6,0) circle (0.12); \draw[] (7,1) circle (0.12);

\draw[] (0,2) circle (0.12); \draw[] (1,3) circle (0.12); \draw[] (2,4) circle (0.12);
\draw[] (3,4) circle (0.12); \draw[] (4,2) circle (0.12); \draw[] (5,3) circle (0.12);
\draw[] (6,4) circle (0.12); \draw[] (7,4) circle (0.12);

\draw[] (0,3) circle (0.12); \draw[] (1,2) circle (0.12); \draw[] (2,1) circle (0.12);
\draw[] (3,3) circle (0.12); \draw[] (4,3) circle (0.12); \draw[] (5,2) circle (0.12);
\draw[] (6,1) circle (0.12); \draw[] (7,3) circle (0.12);


\draw[-stealth,shorten >=6pt,color=gray] (0,0) to[bend left=30] (0,2);
\draw[-stealth,shorten >=6pt,color=gray] (0,2) to[bend left=30] (0,3);
\draw[-stealth,shorten >=6pt,color=gray] (0,3) to[bend left=19] (0,0);

\draw[-stealth,shorten >=6pt,color=gray] (1,1) to[bend left=22] (1,3);
\draw[-stealth,shorten >=6pt,color=gray] (1,3) to[bend left=30] (1,2);
\draw[-stealth,shorten >=6pt,color=gray] (1,2) to[bend left=30] (1,1);

\draw[-stealth,shorten >=6pt,color=gray] (2,0) to[bend left=16] (2,4);
\draw[-stealth,shorten >=6pt,color=gray] (2,4) to[bend left=19] (2,1);
\draw[-stealth,shorten >=6pt,color=gray] (2,1) to[bend left=30] (2,0);

\draw[-stealth,shorten >=6pt,color=gray] (3,1) to[bend left=19] (3,4);
\draw[-stealth,shorten >=6pt,color=gray] (3,4) to[bend left=30] (3,3);
\draw[-stealth,shorten >=6pt,color=gray] (3,3) to[bend left=22] (3,1);

\draw[-stealth,shorten >=6pt,color=gray] (4,0) to[bend left=30] (4,2);
\draw[-stealth,shorten >=6pt,color=gray] (4,2) to[bend left=30] (4,3);
\draw[-stealth,shorten >=6pt,color=gray] (4,3) to[bend left=19] (4,0);

\draw[-stealth,shorten >=6pt,color=gray] (5,1) to[bend left=22] (5,3);
\draw[-stealth,shorten >=6pt,color=gray] (5,3) to[bend left=30] (5,2);
\draw[-stealth,shorten >=6pt,color=gray] (5,2) to[bend left=30] (5,1);

\draw[-stealth,shorten >=6pt,color=gray] (6,0) to[bend left=16] (6,4);
\draw[-stealth,shorten >=6pt,color=gray] (6,4) to[bend left=19] (6,1);
\draw[-stealth,shorten >=6pt,color=gray] (6,1) to[bend left=30] (6,0);

\draw[-stealth,shorten >=6pt,color=gray] (7,1) to[bend left=19] (7,4);
\draw[-stealth,shorten >=6pt,color=gray] (7,4) to[bend left=30] (7,3);
\draw[-stealth,shorten >=6pt,color=gray] (7,3) to[bend left=22] (7,1);

\end{tikzpicture}
\end{center}

\caption{$00$-controlled  and  $01$-controlled rotations built from $0$-controlled rotations. These are controlled by the two bottommost wires, the rotation rotates the wires in order $2 \rightarrow 3 \rightarrow 4 \rightarrow 2$, where the bottommost wire is the $0$th. The diagram is read from left to right, and in each column we perform a $0$-controlled rotation. A black circle indicates a control wire, and white circles are the rotated wires, and the arrows indicate the direction of rotation. 
}
\label{fig:ccrot}
\end{figure}

A similar brute force search also shows the following (again six is optimal).

\begin{lemma}
\label{lem:SecondLemmaOfThisSection}
The word cycle $(0001 \; 0010 \; 0100)$ can be built from six $0$-word-controlled three-wire rotations. The same is true for $(0011 \; 0110 \; 0101)$.
\end{lemma}



Let $\pi_1 = (001 \; 010 \; 100)$ and $\pi_2 = (011 \; 110 \; 101)$. Note that $\pi_1 \circ \pi_2$ is the three-wire rotation. Then, by Lemma~\ref{lem:FirstLemmaOfThisSection} and Lemma~\ref{lem:InductionLemma}, 
$1$-control $(\pi_1 \circ \pi_2)$-permutations generate $k$-controlled $(\pi_1 \circ \pi_2)$-permutations for all $k$. By 
Lemma~\ref{lem:SecondLemmaOfThisSection}, $1$-controlled $(\pi_1 \circ \pi_2)$-permutations generate $1$-controlled $\{\pi_1, \pi_2\}$-permutations, and then by Lemma~\ref{lem:ExtraWireLemma}, $k$-controlled $(\pi_1 \circ \pi_2)$-permutations generate $k$-controlled $\{\pi_1, \pi_2\}$-permutations. 
Putting these together and combining with Corollary~\ref{cor:Something}, we have:

\begin{theorem}
Let $A = \{0,1\}$. Then the alternating conservative \what{} $\Cfour$ is generated by the controlled wire rotation
\[ f(a,b,c,d) = \left\{\begin{array}{cc}
(a,c,d,b) & \mbox{if } a = 0 \\
(a,b,c,d) & \mbox{otherwise}
\end{array}\right. \]
and the even conservative permutations of $A^3$.
\end{theorem}

Clearly $f(a,b,c,d)$ is generated by $1$-controlled wire swaps. It follows that the Fredkin gate together with the (unconditional) wire swap generates all even conservative permutations of $\{0,1\}^n$ for $n \geq 4$.





\section{Conclusion}


We have precisely determined the \what{} generated by a finite set of generators on an even order alphabet and show that on an odd alphabet, a finite collection of mappings generates the whole \what{}. The first result confirms a conjecture in \cite{boykett15} and the second gives a simpler proof of the same result from that paper. Moreover, we have shown that the alternating conservative \what{} is finitely generated on all alphabets, but the conservative \what{} is never finitely generated.

The methods are rather general: We have developed an induction result (Lemma~\ref{lem:InductionLemma}) for finding generating sets for \what{}s of controlled permutations, allowing us to determine finite generating sets for some \what{}s with uniform methods. We also prove the nonexistence of a finite generating family for conserved gates with a general method in Theorem~\ref{thm:ConservedQuantityNonFG}, when borrowed bits are not used. We only need particular properties of the weight function in the proof of Theorem~\ref{thm:XuGeneralization}, where it is shown that the (usual) conservative \what{} is not finitely generated even when borrowed bits are allowed.



While this paper develops strong techniques for showing finitely generatedness and non-finitely generatedness of \what{}s, our generating sets are also somewhat abstract,  not corresponding very well to known generating sets. It would be of value to replace the constructions found in section \ref{secsearches} by more understandable constructions, in order to find more concrete generating sets in the case of general alphabets for conservative gates. 

\bibliographystyle{elsarticle-num}
\bibliography{revcompproc}

\def\cprime{$'$}
\begin{thebibliography}{10}
\providecommand{\url}[1]{\texttt{#1}}
\providecommand{\urlprefix}{URL }

\bibitem{aaronsonetal15}
Aaronson, S., Grier, D., Schaeffer, L.: The classification of reversible bit
  operations. Electronic Colloquium on Computational Complexity (66) (2015)

\bibitem{boykett15}
Boykett, T.: Closed systems of invertible maps (2015),
  \url{http://arxiv.org/abs/1512.06813}, submitted

\bibitem{rc2016}
Boykett, T., Kari, J., Salo, V.: Strongly universal reversible gate sets. In:
  Devitt, S.J., Lanese, I. (eds.) Reversible Computation - 8th International
  Conference, {RC} 2016, Bologna, Italy, July 7-8, 2016, Proceedings. Lecture
  Notes in Computer Science, vol. 9720, pp. 239--254. Springer (2016),
  \url{http://dx.doi.org/10.1007/978-3-319-40578-0_18}

\bibitem{FrTo82}
Fredkin, E., Toffoli, T.: Conservative logic. International Journal of
  Theoretical Physics  21(3),  219--253 (1982),
  \url{http://dx.doi.org/10.1007/BF01857727}

\bibitem{lafont93}
LaFont, Y.: Towards an algebraic theory of boolean circuits. Journal of Pure
  and Applied Algebra  184,  257--310 (2003)

\bibitem{musset97}
Musset, J.: G\'en\'erateurs et relations pour les circuits bool\'eens
  r\'eversibles. Tech. Rep. 97-32, Institut de Math\'ematiques de Luminy
  (1997), \url{http://iml.univ-mrs.fr/editions/}

\bibitem{RoGa99}
Rosales, J., Garc{\'\i}a-S{\'a}nchez, P.: Finitely Generated Commutative
  Monoids. Nova Science Publishers (1999),
  \url{https://books.google.fi/books?id=LQsH6m-x8ysC}

\bibitem{selinger16}
Selinger, P.: Reversible k-ary logic circuits are finitely generated for odd k
  (April 2016), arXiv

\bibitem{szendrei}
Szendrei, {\'A}.: Clones in universal algebra, S\'eminaire de Math\'ematiques
  Sup\'erieures [Seminar on Higher Mathematics], vol.~99. Presses de
  l'Universit\'e de Montr\'eal, Montreal, QC (1986)

\bibitem{toff80}
Toffoli, T.: Reversible computing. Tech. Rep. MIT/LCS/TM-151, MIT (1980)

\bibitem{xu15}
Xu, S.: Reversible Logic Synthesis with Minimal Usage of Ancilla Bits. Master's
  thesis, MIT (June 2015), \url{http://arxiv.org/pdf/1506.03777.pdf}

\bibitem{yangetal05}
Yang, G., Song, X., Perkowski, M., Wu, J.: Realizing ternary quantum switching
  networks without ancilla bits. J. Phys. A  38(44),  9689--9697 (2005),
  \url{http://dx.doi.org/10.1088/0305-4470/38/44/006}

\end{thebibliography}

\end{document}